\pdfoutput=1

\documentclass[journal]{IEEEtran}
\makeatletter
\if@twocolumn
  \newcommand{\bildweite}{\columnwidth}%

\else% \@twocolumnfalse, onecolumn
  \newcommand{\bildweite}{0.5\columnwidth}%

\fi
\makeatother

\usepackage{cite}
\ifCLASSINFOpdf
  \usepackage[pdftex]{graphicx}
  \graphicspath{{./Sections/}{./Figures/}{../Figures/pdf/}{../Figures/jpeg/}}
\else
  \usepackage[dvips]{graphicx}
  \graphicspath{{../eps/}}
\fi
\usepackage[cmex10]{amsmath}
\usepackage{theorem}

\usepackage{latexsym}                % wegen \Box
\usepackage{dsfont}                  % wegen mathematische notation von mengensymbole
\usepackage{mathtools}               % wegen erweiterte mathematische notation, vorallem aber f\"{u}r das mathematische definitionszeichen
\usepackage{xcolor}                  % wegen incscape, um gute grafiken zu erstellen
\usepackage{xfrac}                   % wegen 1/2 zeichen, br\"{u}che mit \sfrac[fontname]{numerator}[/]{denominator}
\usepackage{abkuerzungen}            % wegen verschiedene Abk\"{u}rzungen
\usepackage[]{hyperref}
\hypersetup{
    pdfpagemode     = UseNone,
    pdfstartpage    = 1,                                      % Startseite
    pdfmenubar      = false,                                    % show Acrobat's menu?
    pdfview         = {FitV},                                      % fits the height of the page to the window
    pdfstartview    = {FitV},                                 % fits the height of the page to the window
    pdfpagelayout   = SinglePage,
    pdffitwindow    = true,                                   % window fit to page when opened
    pdfcenterwindow = true,
}

\begin{document}

\title{A Sharp Double Inequality for\\the Inverse Tangent Function}

%\author{Gholamreza~Alirezaei,~\IEEEmembership{Member,~IEEE,}
%        and~Rudolf~Mathar,~\IEEEmembership{Member,~IEEE}% <-this % stops a space
\author{Gholamreza~Alirezaei% <-this % stops a space
%\author{Gholamreza~Alirezaei,~\IEEEmembership{Member,~IEEE}% <-this % stops a space

\thanks{G.~Alirezaei is with the Institute for Theoretical Information Technology,
RWTH Aachen University, 52056 Aachen, Germany (e-mail: alirezaei@ti.rwthaachen.de).}% <-this % stops a space
\thanks{The present work is categorized in terms of Mathematics Subject Classification (MSC2010): 26D05, 26D07, 26D15, 33B10, 39B62.}}
\maketitle

\begin{abstract}
%\boldmath
  The inverse tangent function can be bounded by different inequalities, for example by Shafer's inequality. In this publication, we propose a new sharp double inequality, consisting of a lower and an upper bound, for the inverse tangent function. In particular, we sharpen Shafer's inequality and calculate the best corresponding constants. The maximum relative errors of the obtained bounds are approximately smaller than ${0.27\%}$ and ${0.23\%}$ for the lower and upper bound, respectively. Furthermore, we determine an upper bound on the relative errors of the proposed bounds in order to describe their tightness analytically. Moreover, some important properties of the obtained bounds are discussed in order to describe their behavior and achieved accuracy.

\end{abstract}

\begin{IEEEkeywords}
  trigonometric bounds; Shafer's inequality; inverse tangent approximation;
\end{IEEEkeywords}

% \fi
\IEEEpeerreviewmaketitle

\section{Introduction}
\label{sec:Introduction}
\IEEEPARstart{T}{he} inverse tangent function is an elementary mathematical function that appears in many applications, especially in different fields of engineering. In electrical engineering, especially in the communication theory and signal processing, it is mostly used to describe the phase of a complex-valued signal. But there are many other applications in which the inverse tangent function plays an important role. On the one hand, it is often used as an approximation for more complex functions because of its elementary behavior. For instance, the Heaviside step function is the most famous function that can be very accurately approximated by the inverse tangent function. On the other hand, it is sometimes approximated by simpler functions in order to enable further calculations. For instance, the inverse tangent function can be accurately approximated by its argument if the absolute value of the argument is sufficiently small. Quite naturally the problem arises how to replace the inverse tangent function with a surrogate function, in order to approximate the inverse tangent function as well as other contemplable functions accurately. If a surrogate function with a mathematically simple form could be found, then the subsequent application of such a surrogate function would be considerable. A few application cases are in the field of information and estimation theory where an unknown phase shift or the direction of arrival is estimated, for example by using the CORDIC-algorithm~\cite{Volder}, the MUSIC-algorithm~\cite{Schmidt}, or MAP and ML estimators~\cite{Fu}. Some other cases are in the field of system design and control theory where a non-linear network unit is modeled by a non-linear function, for instance the saturation behavior of an amplifier~\cite{Zeng} or the sigmoidal non-linearity in neuronal networks~\cite{Widrow}. Some more application cases are related to the theory of signals and systems, where a signal should be mapped into a set of coefficients of basis functions, however the transformation is not feasible because of the phase description by the inverse tangent function, for example in some Fourier-related transforms~\cite{Hwang}. Many other applications are likewise conceivable.

But we have to mention that finding a simple replacement for the inverse tangent function is in fact difficult. In the present work, we thus focus only on a special idea which has some nice properties and is described in the following.

In~\cite{Shafer1966}, R.~E.~Shafer proposed the elementary problem: \emph{Show that for all ${x > 0}$ the inequality
\begin{equation}\label{eq:ShaferInequality}
   \frac{3 x}{1+2 \sqrt{1+x^2}} <   \arctan(x)
\end{equation}
holds}, where $\arctan(x)$ denotes the inverse tangent function that is defined for all real numbers $x$. From Shafar's problem several inequalities have been emerged to date. In particular, the authors in~\cite{Qi2009} investigated double inequalities of the form
\begin{equation}\label{eq:QiInequality1}
   \frac{a_1 x}{a_2 + \sqrt{1+x^2}}  <   \arctan(x)  <  \frac{b_1 x}{b_2 + \sqrt{1+x^2}} \ , \ \ x > 0 \ ,
\end{equation}
and they determined the coefficients $a_1$, $a_2$, $b_1$ and $b_2$ such that the above double inequality is sharp. In the present work, we follow a similar idea and investigate a generalized version of the double inequality in~\eqref{eq:QiInequality1}. We consider functions of the type
\begin{equation}\label{eq:functiontype}
   \frac{x}{c_1 + \sqrt{c_2 + c_3 x^2}}
\end{equation}
with positive real coefficients $c_1$, $c_2$ and $c_3$, because such kind of functions has advantageous properties in order to replace the inverse tangent function as we will discuss in the next section. Then we determine the triple ${(c_1, c_2, c_3)}$ such that a lower and an upper bound for the inverse tangent function is achieved. In order to describe the tightness of the obtained bounds, we determine an upper bound on the relative errors of the proposed bounds. Furthermore, we discuss some corresponding properties of the proposed bounds and visualize the achieved results.

\hfill G.~A.

\hfill July 18, 2013

\subsection*{Mathematical Notations:}
Throughout this paper we denote the set of real numbers by $\setR$. The mathematical operation $\abs{x}$ denotes the absolute value of any real number $x$. Furthermore, ${\Ord\bigl(\omega(x)\bigr)}$ denotes the order of any function $\omega(x)$.

\section{Main Theorems}
\label{sec:Main_Theorems}
In the current section, we present the new bounds for the inverse tangent function and describe some of their important properties.
\begin{theorem}\label{thm:SharpShaferInequality}
      For all ${x \in \setR}$, let ${f(x)}$, ${g(x)}$ and ${h(x)}$ be defined by
      \begin{equation}\label{eq:SharpShaferInequalityF}
        f(x) \coloneqq \frac{x}{\frac{4}{\pi^2} + \sqrt{\bigl(1 - \frac{4}{\pi^2}\bigr)^2 + \frac{4 x^2}{\pi^2}}} \ ,
      \end{equation}
      \begin{equation}\label{eq:SharpShaferInequalityG}
        g(x) \coloneqq \arctan(x)
      \end{equation}
      and
      \begin{equation}\label{eq:SharpShaferInequalityH}
        h(x) \coloneqq \frac{x}{ 1 - \frac{6}{\pi^2}  + \sqrt{\bigl(\frac{6}{\pi^2}\bigr)^2 + \frac{4 x^2}{\pi^2}}} \ .
      \end{equation}
      Then, for all ${x \in \setR}$, the double inequality
      \begin{equation}\label{eq:SharpShaferInequality}
        \abs{f(x)} \ \leq \ \abs{g(x)} \ \leq \ \abs{h(x)}
      \end{equation}
      holds.
\end{theorem}
\begin{proof}
  See Appendix~\ref{sec:Appendix_A}.
\end{proof}
\begin{remark}
      The functions $f(x)$, $g(x)$ and $h(x)$ are point symmetric such that the identities $f(-x) = -f(x)$, $g(-x) = -g(x)$ and $h(-x) = -h(x)$ hold. Hence, it is sufficient to consider only the case of ${x \geq 0}$.
\end{remark}
\begin{remark}
      The triples ${\bigl(\frac{4}{\pi^2}, (1 - \frac{4}{\pi^2} )^2, \frac{4}{\pi^2} \bigr)}$ and ${\bigl(1 - \frac{6}{\pi^2}, (\frac{6}{\pi^2} )^2, \frac{4}{\pi^2} \bigr)}$ are the best possible ones such that the above double inequality holds. In other words, no component of the first triple can be replaced by a smaller value and no component of the second triple can be replaced by a larger value with respect to ${x \geq 0}$ while keeping the other components fixed. In this sense, the double inequality in Theorem~\ref{thm:SharpShaferInequality} is sharp.
\end{remark}

We get a first impression of the nature of the bounds from Figure~\ref{fig:result1}. As we can see the double inequality in Theorem~\ref{thm:SharpShaferInequality} is very tight. The curves seem to be continuous, strictly increasing and convex. Hence, we elaborately discuss the mathematical properties of the obtained bounds in the following.
\begin{figure}	% h-here, t-top, b-bottom
  \centering
  \includegraphics[width=\bildweite]{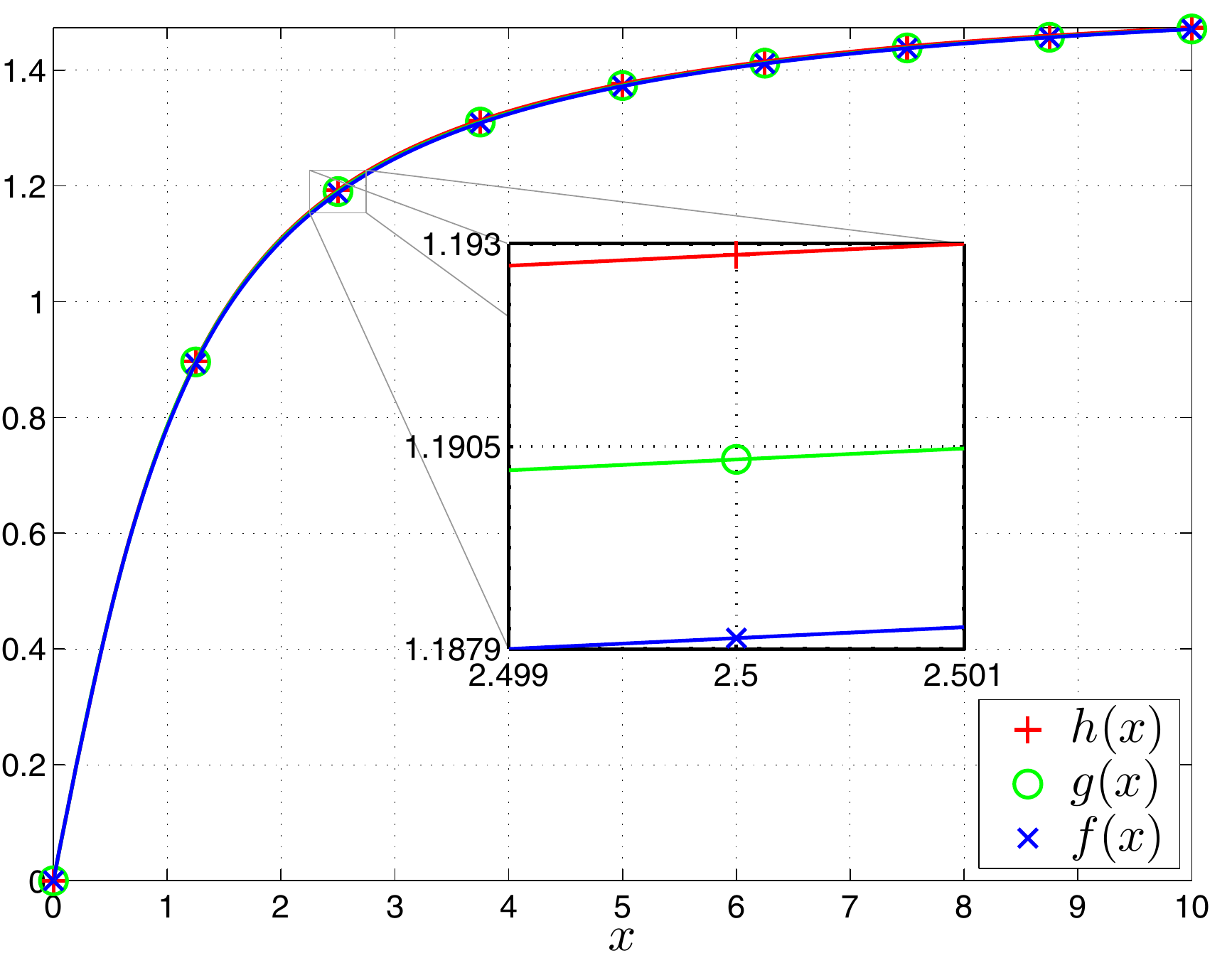}
  \caption{The inverse tangent function and its bounds from Theorem~\ref{thm:SharpShaferInequality} are visualized for the range of ${0 \leq x \leq 10}$. The curves are closely adjacent to one another such that without magnification the differences are not really visible. The curves are equal at zero and approach the same upper limit as $\abs{x}$ approaches infinity.}\label{fig:result1}
\end{figure}

On the one hand, the first three elements in the Taylor series expansions of ${f(x)}$, ${g(x)}$ and ${h(x)}$ as $\abs{x}$ approaches zero are obtained as
      \begin{equation}\label{eq:TaylerExpansionF}
        f(x) \simeq x - \frac{2}{\pi^2-4}\, x^3 + 2 \frac{3 \pi^2 -8}{(\pi^2-4)^3}\, x^5 + \Ord\bigl(x^7\bigr)\ ,
      \end{equation}
      \begin{equation}\label{eq:TaylerExpansionG}
        g(x) \simeq x - \frac{1}{3}\, x^3 + \frac{1}{5}\, x^5 + \Ord\bigl(x^7\bigr)
      \end{equation}
      and
      \begin{equation}\label{eq:TaylerExpansionH}
        h(x) \simeq x - \frac{1}{3}\, x^3 + \frac{\pi^2 +12}{108}\, x^5  + \Ord\bigl(x^7\bigr) \ .
      \end{equation}
\begin{remark}\label{rem:TaylerExpansion}
      Only the both first elements in the Taylor series expansions of ${f(x)}$ and ${g(x)}$ are identical to each other, while in the Taylor series expansions of ${h(x)}$ and ${g(x)}$ the both first two elements are pairwise identical to each other. Thus, ${h(x)}$ achieves a better approximation of ${g(x)}$ than ${f(x)}$ for sufficiently small $\abs{x}$.
\end{remark}
On the other hand, the first three elements in the asymptotic power series expansions of ${f(x)}$, ${g(x)}$ and ${h(x)}$ as $\abs{x}$ approaches infinity are obtained as
      \begin{equation}\label{eq:AsymptoticExpansionF}
        f(x) \simeq \frac{\pi}{2} - \frac{1}{x} - \frac{\pi^4-8 \pi^2 -16}{16 \pi x^2}  + \Ord\bigl(x^{-3}\bigr)\ ,
      \end{equation}
      \begin{equation}\label{eq:AsymptoticExpansionG}
        g(x) \simeq \frac{\pi}{2} - \frac{1}{x} + \frac{1}{3 x^3} + \Ord\bigl(x^{-5}\bigr)
      \end{equation}
      and
      \begin{equation}\label{eq:AsymptoticExpansionH}
        h(x) \simeq \frac{\pi}{2} - \frac{\pi^2 - 6}{4 x} + \frac{\pi^4- 12 \pi^2 +18}{8 \pi x^2} + \Ord\bigl(x^{-3}\bigr)
      \end{equation}
by using the general definition of the asymptotic power series expansion~\cite[p.~11, Definition~1.3.3]{Bleistein} and simple calculations.
\begin{remark}\label{rem:AsymptoticExpansion}
      The both first two elements in the asymptotic power series expansions of ${f(x)}$ and ${g(x)}$ are pairwise identical to each other while in the asymptotic power series expansions of ${h(x)}$ and ${g(x)}$ only the both first elements are identical to each other. Thus, ${f(x)}$ achieves a better approximation of ${g(x)}$ than ${h(x)}$ for sufficiently large $\abs{x}$.
\end{remark}
\begin{corollary}\label{cor:SharpShaferIdentity}
      From equations~\eqref{eq:TaylerExpansionF}--\eqref{eq:AsymptoticExpansionH} we conclude that
      \begin{equation}\label{eq:SharpShaferIdentity}
        \lim\limits_{x \mapsto \pm 0} f(x) = \lim\limits_{x \mapsto \pm 0} g(x) = \lim\limits_{x \mapsto \pm 0} h(x) = 0
      \end{equation}
      and
      \begin{equation}\label{eq:SharpShaferIdentity}
        \lim\limits_{x \mapsto \pm \infty} f(x) = \lim\limits_{x \mapsto \pm \infty} g(x) = \lim\limits_{x \mapsto \pm \infty} h(x) = \pm\frac{\pi}{2} \ .
      \end{equation}
\end{corollary}
\begin{lemma}\label{lem:ContinuityOfFH}
  For all ${x \in \setR}$, both bounds $f(x)$ and $h(x)$ are continuous.
\end{lemma}
\begin{proof}
  Both numerators and denominators of $f(x)$ and $h(x)$ are continuous functions in $x$ and the denominators are always non-zero which imply the absence of discontinuities.
\end{proof}
\begin{lemma}\label{lem:MonotonicityOfFH}
  For all ${x \in \setR}$, both bounds $f(x)$ and $h(x)$ are strictly increasing.
\end{lemma}
\begin{proof}
  By differentiation we obtain the following first derivative
  \begin{equation}\label{eq:first_derivative}
    \frac{\mathrm{d}}{\dx} \frac{x}{c_1 + \sqrt{c_2 + c_3 x^2}} =  \frac{c_2 + c_1 \sqrt{c_2 + c_3 x^2}}{\sqrt{c_2 + c_3 x^2}\, \bigl[c_1 + \sqrt{c_2 + c_3 x^2}\, \bigr]^2} \ .
  \end{equation}
  This derivative is positive for all ${x \in \setR}$ because $c_1$, $c_2$ and $c_3$ are positive constants in both bounds. Hence, the bounds are strictly increasing.
\end{proof}
\begin{corollary}
  Both bounds $f(x)$ and $h(x)$ are differentiable on $\setR$, because the first derivative of the bounds exists due to the derivative in equation~\eqref{eq:first_derivative}.
\end{corollary}
\begin{corollary}
  Both bounds $f(x)$ and $h(x)$ are limited, due to equation~\eqref{eq:SharpShaferIdentity} and because of the monotonicity in Lemma~\ref{lem:MonotonicityOfFH}.
\end{corollary}
\begin{corollary}
  Both bounds $f(x)$ and $h(x)$ do not have any critical points, because they are strictly increasing and differentiable on $\setR$.
\end{corollary}
\begin{lemma}\label{lem:ConvexityOfFH}
  For all ${x \geq 0}$, both bounds $f(x)$ and $h(x)$ are concave while for all ${x \leq 0}$, both bounds are convex.
\end{lemma}
\begin{proof}
  By differentiation we obtain the following second derivative
  \begin{multline}
    \frac{\mathrm{d}^2}{\dx^2} \frac{x}{c_1 + \sqrt{c_2 + c_3 x^2}} \\
    = - c_3\, x\,  \frac{3 c_1 c_2 + 2 c_1 c_3 x^2 + 3 c_2 \sqrt{c_2 + c_3 x^2}}{( c_2 + c_3 x^2 )^{\sfrac{3}{2}}\, \bigl[c_1 + \sqrt{c_2 + c_3 x^2}\, \bigr]^3} \ .
  \end{multline}
  The sign of this derivative is only dependent on $x$ because $c_1$, $c_2$ and $c_3$ are positive constants in both bounds. Hence, this derivative is non-positive for all ${x \geq 0}$ and non-negative for all ${x \leq 0}$ which completes the proof.
\end{proof}
\begin{corollary}
  Both bounds $f(x)$ and $h(x)$ have the same unique inflection point at the origin, due to opposing convexities for ${x \geq 0}$ and ${x \leq 0}$, see Lemma~\ref{lem:ConvexityOfFH}.
\end{corollary}

In the following enumeration, we now summarize the properties of the bounds that have been shown, thus far.
\begin{enumerate}
  \item The bounds are equal only at zero and they approach the same limit as $\abs{x}$ approaches infinity.
  \item Both bounds are point symmetric, continuous, strictly increasing, differentiable, and limited.
  \item They are convex for all ${x\leq0}$ and concave otherwise.
  \item There are no critical points.
  \item Both bounds have the same unique inflection point.
\end{enumerate}

The above properties enable us to use the proposed bounds suitable in future works. It remains to show the tightness of the bounds with respect to the inverse tangent function. For this purpose we deduce an upper-bound on the actual relative errors of the bounds, in the following.
\begin{definition}\label{def:RelativeError}
  For all ${x \in \setR}$, the relative errors of the bounds given in Theorem~\ref{thm:SharpShaferInequality} are defined by
  \begin{equation}\label{eq:DefRelativeErrorLower}
    r_f(x) \coloneqq \frac{g(x) - f(x)}{g(x)}
  \end{equation}
  and
  \begin{equation}\label{eq:DefRelativeErrorUpper}
    r_h(x) \coloneqq \frac{h(x) - g(x)}{g(x)} \ .
  \end{equation}
\end{definition}

Note that ${x=0}$ is a removable singularity for both last ratios because of approximations~\eqref{eq:TaylerExpansionF},~\eqref{eq:TaylerExpansionG} and~\eqref{eq:TaylerExpansionH}. Thus, $r_f(x)$ and $r_h(x)$ are continuously extendable over ${x=0}$. All following fractions are also continuously extendable over ${x=0}$ in a similar manner so that no difficulties related to singularities occur, hereinafter.
\begin{figure}	% h-here, t-top, b-bottom
  \centering
  \includegraphics[width=\bildweite]{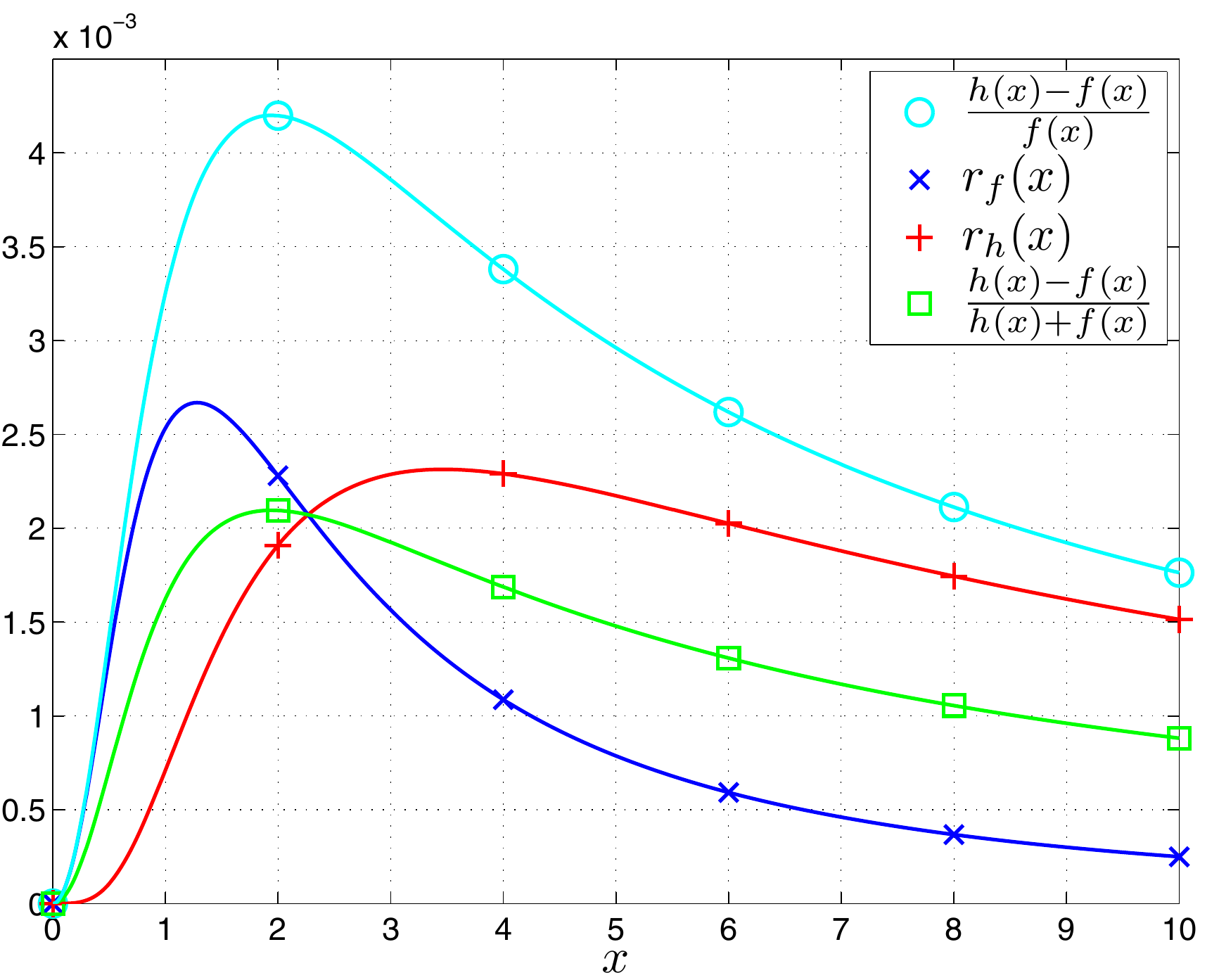}
  \caption{The relative errors of the bounds for the inverse tangent function are visualized for the range of ${0 \leq x \leq 10}$. All curves are equal at zero and they approach  zero as $\abs{x}$ approaches infinity. $r_h(x)$ is smaller, and hence, better than $r_f(x)$ for sufficiently small values of $\abs{x}$, while for sufficiently large values of $\abs{x}$ the opposite holds. The relative error $r_f(x)$ is approximately smaller than $0.27\%$ while $r_h(x)$ is approximately smaller than $0.23\%$. The ratio $[h(x)-f(x)]/[h(x)+f(x)]$ lies between $r_f(x)$ and $r_h(x)$, and hence, is an upper bound for $\min\{r_f(x), r_h(x) \}$ while ${[h(x)-f(x)]/f(x)}$ is an upper bound for $\max\{r_f(x), r_h(x) \}$.}\label{fig:result2}
\end{figure}
\begin{theorem}\label{thm:RelativeError}
      For all ${x \in \setR}$, the inequalities
      \begin{multline}\label{eq:RelativeErrorMax}
        \max\bigl\{r_f(x), r_h(x)\bigr\} \leq \frac{h(x)-f(x)}{f(x)} \\ = \frac{10-\pi^2-2 \sqrt{9+\pi^2 x^2}+\sqrt{\left(\pi^2-4\right)^2 + 4 \pi^2 x^2}}{\pi^2-6+2 \sqrt{9+\pi^2 x^2}}
      \end{multline}
      and
      \begin{multline}\label{eq:RelativeErrorMin}
        \min\bigl\{r_f(x), r_h(x)\bigr\} \leq \frac{h(x)-f(x)}{h(x)+f(x)} \\ = \frac{10-\pi^2-2 \sqrt{9+\pi^2 x^2}+\sqrt{ \left(\pi^2-4\right)^2 + 4 \pi^2 x^2}}{\pi^2-2+2 \sqrt{9+\pi ^2 x^2}+\sqrt{\left(\pi^2-4\right)^2 + 4 \pi^2 x^2}}\\
         \leq \max\bigl\{r_f(x), r_h(x)\bigr\}
      \end{multline}
      hold.
\end{theorem}
\begin{proof}
  See Appendix~\ref{sec:Appendix_B}.
\end{proof}

Note, that the inequalities in Theorem~\ref{thm:RelativeError} do not contain the inverse tangent function, at all.

In Figure~\ref{fig:result2}, the relative errors of the obtained bounds are shown. The maximum relative errors of the bounds are approximately smaller than ${0.27\%}$ and ${0.23\%}$ for $f(x)$ and $h(x)$, respectively. It is worthwhile mentioning that both bounds are valid for the whole domain of real numbers.

\section{Conclusion}
\label{sec:Conclusion}
In the present work, we have investigated the approximation of the inverse tangent function and deduced two new bounds. We have derived a lower and an upper bound with simple closed-form formulae which are sharp and very accurate. Furthermore, we have presented some useful and important properties of the obtained bounds. These properties can be necessary in future works. Moreover, we have investigated the relative errors of the proposed bounds. The corresponding maximum relative errors of the bounds are approximately smaller than $0.27\%$ and $0.23\%$ for the lower bound and upper bound, respectively. These values show that the obtained bounds are very accurate and thus are suitably applicable in the most engineering problems. Finally, we have illustrated some results in order to visualize the achieved gains.

\appendices
\section{Proof of the Bounds}\label{sec:Appendix_A}
\begin{lemma}\label{lem:InequalityForPi}
      The transcendental number $\pi^2$ can be bounded by the double inequality
      \begin{equation}\label{eq:InequalityForPi}
        \frac{29}{3} \ < \ \pi^2 \ < \ 10 \ .
      \end{equation}
\end{lemma}
\begin{proof}\label{proof:InequalityForPi}
      Both bounds are well known for long, see for example~\cite{Neugebauer}. A new proof of the upper bound can be found in~\cite{Elkies}. We here give an elementary proof of the lower bound. The identities ${1 = \sum_{k=1}^{\infty} \frac{1}{k (k+1)}}$ and ${\zeta(2) = \sum_{k=1}^{\infty} \frac{1}{k^2} = \frac{\pi^2}{6}}$, see for example~\cite[p.~8, eq.~0.233.3 and p.~12, eq.~0.244.3]{Gradshteyn}, are used to deduce
      \begin{multline}
        \frac{\pi^2}{6} = 1 + \sum\limits_{k=1}^{\infty} \frac{1}{k^2} - \sum\limits_{k=1}^{\infty} \frac{1}{k (k+1)} = 1 + \sum\limits_{k=1}^{\infty} \frac{1}{k^2 (k+1)}\\
        = \underbrace{1 + \frac{1}{2} + \frac{1}{12} + \frac{1}{36}}_{=\frac{29}{18}} + \sum\limits_{k=4}^{\infty} \frac{1}{k^2 (k+1)} > \frac{29}{18} \ .
      \end{multline}
      Hence ${\frac{29}{3} < \pi^2}$ follows.
\end{proof}
In the following, we denote the differences ${g(x)-f(x)}$ and ${h(x)-g(x)}$ by
      \begin{equation}\label{eq:DifferenceF}
        \Delta_f(x) \coloneqq  g(x)-f(x)
      \end{equation}
      and
      \begin{equation}\label{eq:DifferenceH}
        \Delta_h(x) \coloneqq  h(x)-g(x) \ ,
      \end{equation}
      respectively. From Corollary~\ref{cor:SharpShaferIdentity} it is immediately deduced that
      \begin{multline}\label{eq:DifferenceIdentity}
        \lim\limits_{x \mapsto \pm 0} \Delta_f(x) = \lim\limits_{x \mapsto \pm 0} \Delta_h(x) \\
        = \lim\limits_{x \mapsto \pm \infty} \Delta_f(x) = \lim\limits_{x \mapsto \pm \infty} \Delta_h(x) = 0 \ .
      \end{multline}
      By direct algebra the first derivatives of $\Delta_f(x)$ and $\Delta_h(x)$ are given as
      \begin{multline}\label{eq:DerivativeOfDifferenceF}
        \frac{\mathrm{d} \Delta_f(x)}{\dx} = \frac{1}{1+x^2}
        - \frac{1}{\frac{4}{\pi^2} + \sqrt{\bigl(1 - \frac{4}{\pi^2}\bigr)^2 + \frac{4 x^2}{\pi^2}}}\\
        + \frac{x^2}{\frac{\pi^2}{4} \sqrt{\bigl(1 - \frac{4}{\pi^2}\bigr)^2 + \frac{4 x^2}{\pi^2}}\, \Bigl(\frac{4}{\pi^2} + \sqrt{\bigl(1 - \frac{4}{\pi^2}\bigr)^2 + \frac{4 x^2}{\pi^2}}\Bigr)^2}
      \end{multline}
      and
      \begin{multline}\label{eq:DerivativeOfDifferenceH}
        \frac{\mathrm{d} \Delta_h(x)}{\dx} = -\frac{1}{1+x^2}
        + \frac{1}{ 1 - \frac{6}{\pi^2}  + \sqrt{\bigl(\frac{6}{\pi^2}\bigr)^2 + \frac{4 x^2}{\pi^2}}} \\
        - \frac{x}{\frac{\pi^2}{4} \sqrt{\bigl(\frac{6}{\pi^2}\bigr)^2 + \frac{4 x^2}{\pi^2}}\, \Bigl( 1 - \frac{6}{\pi^2}  + \sqrt{\bigl(\frac{6}{\pi^2}\bigr)^2 + \frac{4 x^2}{\pi^2}}\Bigr)^2} \ ,
      \end{multline}
      respectively.
\begin{corollary}\label{cor:RootsOfDerivativeOfDifference}
      The first derivatives of $\Delta_f(x)$ and $\Delta_h(x)$ vanish only at three real points, namely
      \begin{equation}\label{eq:RootsOfDerivativeOfDifferenceF}
        x_f \in \biggl\{0,\, \pm \frac{(\pi^2-4) \sqrt{-2 \pi^4 +36 \pi^2 - 160}}{\pi^4 - 8 \pi^2 -16} \, \biggr\}
      \end{equation}
      and
      \begin{equation}\label{eq:RootsOfDerivativeOfDifferenceH}
        x_h \in \biggl\{0,\, \pm \frac{\sqrt{-5 \pi^4 + 108 \pi^2 - 576}}{\pi (10 - \pi^2)} \, \biggr\} \ ,
      \end{equation}
      respectively.
\end{corollary}
\begin{proof}\label{proof:DerivativeOfDifference}
      We set~\eqref{eq:DerivativeOfDifferenceF} and~\eqref{eq:DerivativeOfDifferenceH} equal to zero and obtain the points in~\eqref{eq:RootsOfDerivativeOfDifferenceF} and~\eqref{eq:RootsOfDerivativeOfDifferenceH} by direct calculations. It remains to prove that all points are real. This is done by showing that the discriminant functions \begin{equation}
        y_f(\nu) \coloneqq -2 \nu^4 +36 \nu^2 - 160 = 2 (\nu^2 - 8)(10 - \nu^2)
      \end{equation}
      and
      \begin{equation}
        y_h(\nu) \coloneqq -5 \nu^4 + 108 \nu^2 - 576 = (5 \nu^2 - 48)(12 - \nu^2)
      \end{equation}
      are non-negative for ${\nu=\pi}$. A curve tracing of $y_f(\nu)$ and $y_h(\nu)$ leads to the relationships
      \begin{equation}\label{eq:CurveSketchingF}
        y_f(\nu) \geq 0 \quad\Leftrightarrow\quad 8 \leq \nu^2 \leq 10
      \end{equation}
      and
      \begin{equation}\label{eq:CurveSketchingH}
        y_h(\nu) \geq 0 \quad\Leftrightarrow\quad \frac{48}{5} \leq \nu^2 \leq 12 \ ,
      \end{equation}
      respectively. Hence, both $y_f(\nu)$ and $y_h(\nu)$ are non-negative for all ${\frac{48}{5} \leq \nu^2 \leq 10}$. By comparing the latter double inequality with the double inequality in Lemma~\ref{lem:InequalityForPi} we deduce that $y_f(\pi)$ and $y_h(\pi)$ are non-negative, and hence, all roots in~\eqref{eq:RootsOfDerivativeOfDifferenceF} and~\eqref{eq:RootsOfDerivativeOfDifferenceH} are real.
\end{proof}
\begin{corollary}\label{cor:PositivityOfF}
      The difference $\Delta_f(x)$ is positive for all sufficiently small positive real numbers $x$. For all negative real numbers $x$ with sufficiently small absolute value, the difference $\Delta_f(x)$ is negative.
\end{corollary}
\begin{proof}\label{proof:PositivityOfF}
      We incorporate the equations~\eqref{eq:TaylerExpansionF} and~\eqref{eq:TaylerExpansionG} into~\eqref{eq:DifferenceF} to derive the first-order approximation of $\Delta_f(x)$ as
      \begin{equation}\label{eq:FirstOrderApproximationDf}
        \Delta_f(x) \simeq  \frac{10-\pi^2}{3 (\pi^2 - 4)}\, x^3 + \Ord\bigl(x^{5}\bigr)
      \end{equation}
      for all sufficiently small values of $\abs{x}$. From the double inequality in Lemma~\ref{lem:InequalityForPi}, we deduce that the last ratio is always positive which completes the proof.
\end{proof}
\begin{corollary}\label{cor:PositivityOfH}
      The difference $\Delta_h(x)$ is positive for all sufficiently large positive real numbers $x$. For all negative real numbers $x$ with sufficiently large absolute value, the difference $\Delta_h(x)$ is negative.
\end{corollary}
\begin{proof}\label{proof:PositivityOfH}
      We incorporate the equations~\eqref{eq:AsymptoticExpansionG} and~\eqref{eq:AsymptoticExpansionH} into~\eqref{eq:DifferenceH} to derive the first-order asymptotic approximation of $\Delta_h(x)$ as
      \begin{equation}\label{eq:FirstOrderApproximationDh}
        \Delta_h(x) \simeq  \frac{10-\pi^2}{4}\, x^{-1} + \Ord\bigl(x^{-2}\bigr)
      \end{equation}
      for all sufficiently large values of $\abs{x}$. From the double inequality in Lemma~\ref{lem:InequalityForPi}, we deduce that the last ratio is always positive which completes the proof.
\end{proof}
\begin{proof}[Proof of \textsc{Theorem}~\ref{thm:SharpShaferInequality}]\label{proof:SharpShaferInequality}
      We only consider the case of ${x \geq 0}$. The case of ${x \leq 0}$ can be proved analogously, due to the point symmetric property of all functions in Theorem~\ref{thm:SharpShaferInequality}. On the one hand, we know from Corollary~\ref{cor:RootsOfDerivativeOfDifference} that each of differences $\Delta_f(x)$ and $\Delta_h(x)$ has only one stationary point for all ${x > 0}$. On the other hand, each of them attains equal values at ${x=0}$ and as ${x \mapsto \infty}$, i.e., ${\Delta_f(0) = \Delta_f(x\mapsto \infty) = 0}$ and ${\Delta_h(0) = \Delta_h(x\mapsto \infty) = 0}$, according to the equation~\eqref{eq:DifferenceIdentity}. Hence, and because of Corollary~\ref{cor:PositivityOfF} and~\ref{cor:PositivityOfH}, as $x$ increases from zero to infinity each of the differences $\Delta_f(x)$ and $\Delta_h(x)$ increases monotonically from zero to a maximum value and from there on decreases monotonically toward zero. Thus, both differences $\Delta_f(x)$ and $\Delta_h(x)$ are non-negative for all ${x \geq 0}$. In other words, if one of the differences had at least one sign change for some value of ${x > 0}$, then it would have at least two stationary points for ${x > 0}$, but this contradicts the curve tracing in Corollary~\ref{cor:RootsOfDerivativeOfDifference}.
\end{proof}

\section{Proof of the Relative Errors}\label{sec:Appendix_B}
\begin{definition}\label{def:AuxiliarySet}
      Let $r_f(x)$ and $r_h(x)$ be defined as in Definition~\ref{def:RelativeError}. Then, we define two auxiliary sets by
      \begin{equation}\label{eq:AuxiliarySetIf}
        \setI_f \coloneqq  \bigl\{ x\in\setR \mid r_f(x) \geq r_h(x) \bigr\}
      \end{equation}
      and
      \begin{equation}\label{eq:AuxiliarySetIh}
        \setI_h \coloneqq  \bigl\{ x\in\setR \mid r_f(x) < r_h(x) \bigr\} \ .
      \end{equation}
\end{definition}
Note that both sets $\setI_f$ and $\setI_h$ are disjoint and their union is the whole real domain.
\begin{corollary}\label{cor:InequalityG}
      For all ${x \in \setI_f}$, ${x \geq 0}$, the inequality
      \begin{equation}\label{eq:InequalityG1}
        g(x) \geq  \frac{h(x) + f(x)}{2}
      \end{equation}
      holds. If ${x \in \setI_h}$, ${x \geq 0}$, then the inequality
      \begin{equation}\label{eq:InequalityG2}
        g(x) <  \frac{h(x) + f(x)}{2}
      \end{equation}
      holds. In the case of ${x \in \setI_f}$ with ${x < 0}$ and ${x \in \setI_h}$ with ${x < 0}$ the above inequalities are reversed.
\end{corollary}
\begin{proof}\label{proof:InequalityG}
      For all ${x \in \setI_f}$ with ${x \geq 0}$, and from Definition~\ref{def:RelativeError} and~\ref{def:AuxiliarySet} it follows that
      \begin{multline}\label{eq:InequalityGProof1}
        r_f(x) \geq r_h(x) \ \Leftrightarrow\ r_f(x) g(x) \geq r_h(x) g(x) \\
        \Leftrightarrow\  g(x)- f(x) \geq h(x) - g(x) \ \Leftrightarrow\  g(x) \geq  \frac{h(x) + f(x)}{2} \ .
      \end{multline}
      Similarly, for all ${x \in \setI_h}$ with ${x \geq 0}$ it follows that
      \begin{multline}\label{eq:InequalityGProof2}
        r_f(x) < r_h(x) \ \Leftrightarrow\ r_f(x) g(x) < r_h(x) g(x) \\
        \Leftrightarrow\  g(x)- f(x) < h(x) - g(x) \ \Leftrightarrow\  g(x) <  \frac{h(x) + f(x)}{2} \ .
      \end{multline}
      In the case of ${x < 0}$, the functions $f(x)$, $g(x)$ and $h(x)$ are negative, and hence, the inequalities in~\eqref{eq:InequalityG1} and~\eqref{eq:InequalityG2} are reversed.
\end{proof}
\begin{proof}[Proof of \textsc{Theorem}~\ref{thm:RelativeError}]\label{proof:RelativeError}
    The proof of inequality~\eqref{eq:RelativeErrorMax} follows from inequality~\eqref{eq:SharpShaferInequality} and Definition~\ref{def:RelativeError}. It gives
    \begin{equation}
      r_f(x) = \frac{g(x) - f(x)}{g(x)} \leq \frac{h(x) - f(x)}{g(x)} \leq \frac{h(x) - f(x)}{f(x)}
    \end{equation}
    and
    \begin{equation}
      r_h(x) = \frac{h(x) - g(x)}{g(x)} \leq \frac{h(x) - f(x)}{g(x)} \leq \frac{h(x) - f(x)}{f(x)}
    \end{equation}
    which in turn result in
    \begin{equation}
      \max\bigl\{r_f(x), r_h(x)\bigr\} \leq \frac{h(x)-f(x)}{f(x)} \ .
    \end{equation}
    The proof of inequality~\eqref{eq:RelativeErrorMin} follows from Definition~\ref{def:RelativeError} and Corollary~\ref{cor:InequalityG}. For all ${x \in \setI_f}$ with ${x \geq 0}$, it gives
    \begin{equation}
      r_f(x) = 1-\frac{f(x)}{g(x)} \geq 1- \frac{f(x)}{\frac{h(x) + f(x)}{2}} = \frac{h(x) - f(x)}{h(x)+f(x)}
    \end{equation}
    and
    \begin{equation}
      r_h(x) = \frac{h(x)}{g(x)} - 1 \leq \frac{h(x)}{\frac{h(x) + f(x)}{2}} - 1 = \frac{h(x) - f(x)}{h(x)+f(x)}
    \end{equation}
    which in turn result in
    \begin{equation}\label{eq:InequalityHilfe1}
      r_h(x) \leq \frac{h(x) - f(x)}{h(x)+f(x)} \leq r_f(x) \ .
    \end{equation}
    If ${x \in \setI_h}$ with ${x \geq 0}$, then it gives
    \begin{equation}
      r_f(x) = 1-\frac{f(x)}{g(x)} < 1- \frac{f(x)}{\frac{h(x) + f(x)}{2}} = \frac{h(x) - f(x)}{h(x)+f(x)}
    \end{equation}
    and
    \begin{equation}
      r_h(x) = \frac{h(x)}{g(x)} - 1 > \frac{h(x)}{\frac{h(x) + f(x)}{2}} - 1 = \frac{h(x) - f(x)}{h(x)+f(x)}
    \end{equation}
    which in turn result in
    \begin{equation}\label{eq:InequalityHilfe2}
      r_f(x) < \frac{h(x) - f(x)}{h(x)+f(x)} < r_h(x) \ .
    \end{equation}
    From the double inequalities~\eqref{eq:InequalityHilfe1} and~\eqref{eq:InequalityHilfe2}, we deduce that
    \begin{equation}
      \min\bigl\{r_f(x), r_h(x)\bigr\} \leq \frac{h(x)-f(x)}{h(x)+f(x)} \leq \max\bigl\{r_f(x), r_h(x)\bigr\}
    \end{equation}
    for all ${x \geq 0}$. For the case of ${x < 0}$, the proof can be obtained analogously. The identities in~\eqref{eq:RelativeErrorMax} and~\eqref{eq:RelativeErrorMin} arise from straightforward calculations.
\end{proof}

%\newpage
%\section{Proof of the Bounds for Rice-Fading}\label{sec:Appendix_C}
%\input{./Sections/Appendix_C.tex}

%\section{}

\section*{Acknowledgment}

Research described in the present work was supervised by Univ.-Prof.~Dr.~rer.~nat.~R.~Mathar, Institute for Theoretical Information Technology, RWTH Aachen University. The author would like to thank him for his professional advice and patience.
\newpage

\ifCLASSOPTIONcaptionsoff
  \newpage
\fi



\bibliographystyle{IEEEtran}
\bibliography{IEEEabrv,./Sections/ArcTan_Inequality}
\end{document}